\documentclass[12pt]{article}
\usepackage{hyperref}

\usepackage{amsmath,amssymb,amstext,amsthm,enumerate}

\newcommand{\p}[1]{\mathbb{P}\left[{#1}\right]}

\theoremstyle{plain}
\newtheorem{theorem}{Theorem}[section]
\newtheorem{lemma}[theorem]{Lemma}

\theoremstyle{definition}
\newtheorem*{definition}{Definition}

\theoremstyle{remark}
\newtheorem*{remark}{Remark}

\theoremstyle{definition}
\newtheorem{model}{Model}[section]

\newcommand{\N}{{\mathbb{N}}}
\newcommand{\Nz}{{\mathbb{N}_0}}
\newcommand{\geo}{\operatorname{Geo}}
\newcommand{\bin}{\operatorname{Bin}}
\newcommand{\depth}{\operatorname{depth}}

\renewcommand{\root}{\aleph}
\newcommand{\outdeg}{\operatorname{out-deg}}
\newcommand{\indeg}{\operatorname{in-deg}}

\title{
Justifying the small-world phenomenon \\
via
random recursive trees}
\author{Abbas Mehrabian\thanks{Supported by the Vanier Canada Graduate Scholarships program.
Most of this work was done while the author was visiting Monash University, Australia.}\\
{\small Department of Combinatorics and Optimization}
\\
{\small University of Waterloo}
\\
{\small \texttt{amehrabi@uwaterloo.ca}}
}
\date{}

\begin{document}

\maketitle

\begin{abstract}
We present a new technique for proving logarithmic upper bounds for diameters of evolving random graph models,
which is based on defining a coupling between random graphs and variants of random recursive trees.
The advantage of the technique is three-fold: it is quite simple and provides short proofs,
it is applicable to a broad variety of models including those incorporating preferential attachment,
and it provides bounds with small constants.
We illustrate this by proving, for the first time, logarithmic upper bounds for the diameters of the following well known models:
the forest fire model, the copying model, the PageRank-based selection model, the Aiello-Chung-Lu models, the generalized linear preference model, directed scale-free graphs, the Cooper-Frieze model,
and random unordered increasing $k$-trees.
Our results shed light on why the small-world phenomenon is observed in so many real-world graphs.
\end{abstract}

\section{Introduction}

`Small-world phenomenon' refers to a striking pattern observed in numerous
real-world graphs: most pairs of vertices are connected by a path whose length is considerably smaller than the size of the graph.
Travers and Milgram~\cite{milgram} in 1969 conducted an experiment in which participants were asked to reach a target person by sending a chain letter. 
The average length of all completed chains
was found to be 6.2, an amazingly small number, hence the phrase `six degrees of separation.'
The Webgraph is a directed graph whose vertices are the static web pages, and there is an edge joining two vertices if there is a hyperlink in the first page pointing to the second page.
Broder, Kumar, Maghoul, Raghavan, Rajagopalan, Stata, Tomkins, and Wiener~\cite{Broderetal} in 1999
crawled about 200 million web pages 
and found that the expected shortest-directed-path distance between two random web pages (when a path exists at all) is 16.18; this figure is 6.83 in the corresponding underlying undirected graph.

Backstrom, Boldi, Rosa, Ugander, and Vigna~\cite{four_degrees_separation} studied the Facebook graph in May 2011, which had about 721 million vertices.
The vertices of this graph are people, and two of them are joined by an edge if they are friends on Facebook.
The diameter of the giant component of this graph was found to be 41, and the average distance between reachable pairs was found to be around 4.74.
For other examples, see, e.g., Tables~1 and~2 in~\cite{tables}, Table~8.1 in~\cite{newman_book} or Table~4 in~\cite{google+}.

Another fascinating observation on many real-world graphs is that their degree sequences are heavy-tailed and almost obey a power law.
As Erd\H{o}s-R\'{e}nyi random graphs do not satisfy this property, in recent years a great deal of research has been built around defining new probabilistic models, aiming at capturing the aforementioned and other properties of real-world graphs (see, e.g., Bonato~\cite[Chapter~4]{Bonato} or Chakrabarti and Faloutsos~\cite[Part~II]{chakrabarti_book}).  
Lots of models have been defined so far, yet very few rigorously analysed.


The \emph{diameter} of an undirected graph is the maximum shortest-path distance between any two vertices.
It is a well known metric quantifying how `small-world' the graph is; informally speaking, it measures how quickly one can get from one `end' of the graph to the other.
The diameter is related to various processes, e.g.\ it is within a constant factor of the memory complexity of the depth-first search algorithm.
Also, it is a natural lower bound for the mixing time of any random walk (\cite[Section~7.1]{mixing_times}) and
the broadcast time of the graph (\cite[Section~3]{broadcast_times}).
Another well studied metric is the \emph{average diameter} of a graph, which is the expected value of the shortest-path distance between two random vertices.
Despite the fact that these are two of the most studied parameters of a network, 
for several models introduced in the literature  the degree sequence has been proved to be power-law, but no sublinear upper bound for the diameter or average diameter is known.

We fill in this gap by presenting a new technique for establishing upper bounds for diameters of certain random graph models,
and demonstrate it by proving logarithmic upper bounds for the diameters of a variety of models, including the following well known ones:
the forest fire model~\cite{forest_fire_journal},
the copying model~\cite{copying_model_def},
the PageRank-based selection model~\cite{pagerank_model_journal},
the Aiello-Chung-Lu models~\cite{acl_model},
the generalized linear preference model~\cite{bu_towsley},
directed scale-free graphs~\cite{directed_scale_free},
the Cooper-Frieze model~\cite{cooper-frieze-model},
and random unordered increasing $k$-trees~\cite{randomktree_def}.
This means that in each of these models, 
for \emph{every} pair $(u,v)$ of vertices there exists a very short $(u,v)$-path, a path connecting $u$ and $v$ whose length is logarithmic in the number of vertices.
This implies, in particular, that the average diameters of these models are logarithmic.
We also prove polylogarithmic upper bounds for the diameter of the preferential attachment model with random initial degrees~\cite{randominitial}
in the case that the initial degrees' distribution has an exponential decay.
Prior to this work no sublinear upper bound was known even for the average diameter of any of these models.
(This claim can quickly be verified by looking at Table~8.2 from the recent monograph~\cite{chakrabarti_book},
or~\cite[Table~III]{Chakrabarti},
or the table in~\cite[p.\ 162]{web_survey}: each cited table contains a summary of known results on the diameter and other properties of several real-world network models.)

This is the first paper that proves logarithmic upper bounds for such a wide range of random graph models. 
Our results shed light on why the small-world phenomenon is observed in so many real-world graphs.
At their core, our arguments are based on the fact that in all  models considered here, there is a sort of `rough uniformity' for the (random) destination of each new link.
Thus, we may expect that for any growing network in which
the endpoints of new links are chosen according to a probability distribution that is `not too biased,' i.e.\ does not greatly favour some vertices over others, the diameter grows at most logarithmically.
We believe this is the primary reason that most real-world graphs are small-world.

From a wider perspective, it would be appealing to have a mathematical theory for characterizing those evolving random graphs which have logarithmic diameters. 
This paper is a fundamental step in building this theory.
The technique developed here gives unified simple proofs for known results, provides lots of new ones,
and will help in proving many of the forthcoming network models are small-world. 
We hope this theory will be developed further to cover other network models, e.g.\ spatial models~\cite{spatial_survey}, as well.


\subsection{Our technique and outline of the paper}

We study \emph{evolving models} (also called \emph{on-line} or \emph{dynamic} models), i.e.\ the graph changes over time according to pre-defined probabilistic rules, 
and we are interested in the long-term structure of this evolving graph.
We assume that in discrete time-steps new vertices and edges appear in the graph, but no deletion occurs.
The goal is to show that the evolving graph at time $n$ has diameter $O(\log n)$ \emph{asymptotically almost surely (a.a.s.)}, that is, with probability tending to 1 as $n$ goes to infinity.
In all models considered here, the Chernoff bound implies that the 
number of vertices at time $n$ is $\Omega(n)$ a.a.s., hence we will conclude that a.a.s.\ the evolving graph has diameter $O(\log n)$ when it has $n$ vertices.


Let us informally explain our technique.
In this section when we write a certain graph/tree has a logarithmic diameter/height, we mean its diameter/height has a logarithmic upper bound.
An important object in this paper is a \emph{random recursive tree}, defined as follows:
there exists a single node at time 0, and in every time-step $t=1,2,\dots$, a new node is born and is joined to a uniformly at random (\emph{u.a.r.}) node of the current tree.
It is known that when this tree has $n$ nodes, a.a.s.\ its height is $\Theta(\log n)$~\cite{random_recursive_trees}.
The technique consists of two main steps:
first, we build a \emph{coupling} between our evolving random graph and some variant of a random recursive tree
in such a way that the diameter of the graph is dominated by a linear function of the height of the tree, and then we prove that a.a.s.\ the tree has a logarithmic height.
The second step is usually straightforward (see Lemma~\ref{lem:multichildren} for an example) and the tricky part is defining the `coupled' tree.
Let us give some examples.

To distinguish between a vertex of the graph and that of the tree, the latter is referred to as a `node'.
For models studied in Section~\ref{sec:uniformly}, namely the forest fire model~\cite{forest_fire_journal}, the copying model~\cite{copying_model_def}, and the PageRank-based selection model~\cite{pagerank_model_journal},
the coupled tree is a random recursive tree with weighted edges, which has the same node set as the vertex set of the graph.
Let us assume that the initial graph has one vertex,
so the tree starts with a single node corresponding to this initial vertex.
These models evolve as follows:
in every time-step a new vertex, say $v$, is born and is joined to some random vertices, say $w_1,\dots,w_d$, in the existing graph in such a way that  for each $j$, vertex $w_j$ has a short distance to a u.a.r.\ vertex $x_j$ of the existing graph.
We let the coupled tree evolve as follows:
a new node $v$ is born and is joined to node $x_1$ in the existing tree, and the weight of the edge $vx_1$ in the tree is set to be the distance between $v$ and $x_1$ in the graph.
Then by induction, the distance in the graph between the initial vertex and  $v$ is at most the weighted distance in the tree between the initial node and  node $v$.
Moreover, by construction, the tree evolves as a weighted random recursive tree.
Finally, examining the distribution of the weights carefully, we prove that a.a.s.\ the obtained evolving tree has a logarithmic weighted height.

We remark that in the argument outlined above, we may  \emph{ignore} the other neighbours $w_2,\dots,w_d$ of the new vertex; only the first edge $vw_1$ is effectively used for bounding the diameter.
This is a repeating phenomenon in our arguments.
An interesting implication is that one can quickly and locally build a spanning tree with logarithmic diameter as the graph evolves.
This might have algorithmic applications.

In Section~\ref{sec:pref} we study models that incorporate preferential attachment.
As a simple example, consider the following evolving rule:
in every time-step, a vertex is chosen using preferential attachment, i.e.\ the probability of choosing a specific vertex is proportional to its degree, then a new vertex is born and is joined to the chosen vertex.
It is easy to observe that sampling a vertex using preferential attachment can be done by choosing a u.a.r.\ endpoint of a u.a.r.\ edge of the graph.
Using this sampling procedure, the evolving rule can be re-stated as follows:
in every time-step, an edge $e$ is sampled
u.a.r., then a random endpoint $w$ of $e$ is chosen,
then a new vertex $v$ is born and is joined to $w$.
One of the main novel ideas in this paper is introducing \emph{edge trees} and employing them in this context.
An \emph{edge} tree is a tree whose nodes correspond to the \emph{edges} of the evolving graph.
We couple the evolving graph with an edge tree, and let the edge tree evolve as follows:
in the corresponding time-step a new node $vw$ is born and is joined to a u.a.r.\ node $e$.
Clearly, the edge tree indeed grows like a random recursive tree as the graph evolves, hence its height can be easily bounded.
Moreover, the constructed coupling implies that the graph's diameter is dominated by a linear function of the tree's height, so we conclude that a.a.s.\ the graph has logarithmic diameter.
If a reader wants to read only one theorem from this paper, Theorem~\ref{thm:pref} should be the one,
which formalizes and generalizes this idea and illustrates the crux of our technique without having too much details.
This theorem states that Model~\ref{def_pref}, a generic model based on the preferential attachment scheme, has logarithmic diameter; the Aiello-Chung-Lu models~\cite{acl_model} and the generalized linear preference model~\cite{bu_towsley} are then proved to be special cases of this model.

In the generalized linear preference model, the probability of choosing a specific vertex is proportional to a linear function, say $ax + b$, of the vertex's degree $x$.
Assuming $a$ and $b$ are even positive integers, we handle this by putting $a$ multiple edges corresponding to each edge, and putting $b/2$ loops at each vertex.
Then choosing a u.a.r.\ endpoint of a u.a.r.\ edge in the new graph corresponds to sampling according to the linear function of the degrees in the old graph.
See Theorem~\ref{thm:remark} for details.
At the end of Section~\ref{sec:pref},
we also analyse the `preferential attachment with random initial degrees' model~\cite{randominitial},
and show that if the initial degrees' distribution has an exponential decay, then a.a.s.\ the generated graph has a polylogarithmic diameter.
This is straightforward to prove using the developed machinery, see Theorem~\ref{thm:parid}.

In Section~\ref{sec:directed} we study the `directed scale-free graphs'~\cite{directed_scale_free}.
The diameter of a directed graph is defined as that of the underlying undirected graph (we follow~\cite{forest_fire_journal} in this regard).
When constructing a graph using this model, one may sample vertices according to linear functions of either out-degrees or in-degrees, 
and the two functions have different constant terms.
To cope with this, we introduce `{headless}' and `{tailless}' edges. 
These are dummy edges in the graph that do not 
play any role in connecting the vertices, but they appear in the tree and their job is just to adjust the selection probabilities.
Details can be found in Theorem~\ref{thm:directed}, which states that a generalized version of directed scale-free graphs has logarithmic diameter.


In Section~\ref{combine} we study the Cooper-Frieze model~\cite{cooper-frieze-model},
which is the most general evolving model known to have a power-law degree sequence.
In this model, the neighbours of a new vertex can be chosen either according to degrees or uniformly at random.
For dealing with this intricacy, we couple with a tree having two types of nodes: some correspond to the vertices, and the others correspond to the edges of the graph.
A multi-typed random recursive tree is obtained, in which at every time-step a new node is born and is joined to a node chosen u.a.r.\ from all nodes of a certain type.
We prove that a.a.s\ this tree has a logarithmic height, 
and by using the coupling's definition we conclude that a.a.s\ the Cooper-Frieze model has a logarithmic diameter (see Theorem~\ref{thm:cooper-frieze}).
For this model, proving that the tree has logarithmic height is actually the harder step.

Finally, in Section~\ref{sec:other} we prove logarithmic upper bounds for three further models:
graphs generated by the pegging process~\cite{pegging_def},
random unordered increasing $k$-trees~\cite{randomktree_def},
and random $k$-Apollonian networks~\cite{high_RANs}.
For the first and last of these, it is already known that a.a.s.\ the diameter is $O(\log n)$, but our approach gives a shorter proof. 

\subsection{Related work}

Surprisingly few results are known about the diameters of evolving random graph models.
Chung and Lu~\cite{coupling_online_offline} defined an evolving (online) and a non-evolving (offline) model.
They state that `The online model is obviously much harder to analyze than the offline model', 
and hence analyse the former by coupling it with the latter, which had been analysed before.
The difficulty of analysing evolving models over non-evolving ones arises perhaps due to the dependencies between edges in the former models.

The evolving model that has attracted the most attention is the \emph{linear preference model}, in which in every step a new vertex is born and is joined to a fixed number of old vertices. 
This is done in such a way that the probability of joining to a given vertex is proportional to a linear function of its degree.
A logarithmic upper bound has been proved for the diameter of this model, and sharper results are known in various special cases~\cite{random_recursive_trees,diameter_preferential_attachment,diameters_pa,vander}.
See the remark before Theorem~\ref{thm:remark} for details.
When the new vertex is joined to exactly one vertex in the existing graph (so the resulting evolving graph is always a tree),
a general technique based on branching processes is developed by Bhamidi~\cite{bhamidi}, 
using which he proved the diameter of a variety of preferential attachment trees is a.a.s\ $\Theta(\log n)$.

Chung and Lu~\cite{coupling_online_offline} used couplings with a non-evolving random graph model to prove that the diameter of a certain growth-deletion model is $\Theta(\log n)$.
On one hand, their model is more general than the models we consider, as they allow vertex and edge deletions,
but on the other hand, their result holds for graphs with at least $\omega (n \log n)$ edges whereas our results covers  graphs with $O(n)$ edges, too.
Moreover, their proof is quite technical and uses general martingale inequalities.
See the remark before Theorem~\ref{thm:remark} for details.

Other evolving models whose diameters have been studied include 
the Fabrikant-Koutsoupias-Papadimitriou model~\cite{fkp_not_powerlaw},
protean graphs~\cite{protean_diameter},
the geometric preferential attachment model~\cite{geometric2,geometric_hybrid},
the spatial preferred attachment model~\cite{spa_typical}, random Apollonian networks~\cite{we_rans_abstract,CF13,istvan},
and random surfer Webgraphs~\cite{we_surfer}.
See~\cite[Section~14]{inhom} and~\cite{vander} for collections of results on diameters of non-evolving models.

Some of the above papers estimate the diameter up to constant factors, or even up to $1+o(1)$ factors.
Our approach gives logarithmic upper bounds that are perhaps not tight, but on the positive side, it is applicable to a broad variety of models, including those incorporating preferential attachment.
Another advantage of our technique is simplicity: all proofs given here are elementary and fairly short, and the only probabilistic tools we use are couplings and the Chernoff bound.
The third advantage of our technique is that the constant factor it gives (hidden in the $O(\log n)$ notion) is typically small: for all the models studied here, the constant is at most 20.

Let us emphasize that we are concerned with upper bounds only and no lower bound for the diameter is proved in this paper.
We believe that for all considered models, at least in the special case when the evolving graph is always a tree,
the diameter is $\Theta(\log n)$.

\subsection{Notation}
In this paper graphs can be directed or undirected, but all trees are undirected.
The distance between two vertices is the number of edges in the shortest path connecting them.
If the graph is directed, the direction of edges is ignored when calculating the distance.
The \emph{diameter} of a graph is the maximum distance between any two vertices.
We will work with (weakly) connected graphs only, so the diameter is always well defined.
Graphs may have parallel edges and loops (note that adding these does not change the diameter).
All considered graphs are finite and rooted, i.e.\ there is a special vertex which is called the root.
The \emph{depth} of a vertex is its distance to the root,
and the \emph{height} of a graph is the maximum depth of its vertices. 
Clearly the diameter is at most twice the height, and we always bound the diameter by bounding the height.
The depth of vertex $v$ in graph $G$ is denoted by $\depth(v,G)$.
All logarithms are in the natural base.
Let us denote 
$\N=\{1,2,\dots\}$,
$\Nz=\{0,1,2,\dots\}$,
and 
$[n]=\{1,2,\dots,n\}$.

A \emph{growing graph} is a sequence $(G_t)_{t=0}^{\infty}$ of random graphs such that $G_t$ is a subgraph of $G_{t+1}$ for all $t\in\Nz$.
We always assume that $G_0$ has size $O(1)$.
A \emph{growing tree} is defined similarly.
This sequence can be thought of as a graph `growing' as time passes, and $G_t$ is the state of the graph at time $t$.
We write informal sentences such as `at time $t$, a new vertex is born and is joined to a random vertex of the existing graph,'
which formally means `$G_t$ is obtained from $G_{t-1}$ by adding a new vertex and joining it to a random vertex of $G_{t-1}$.'

\section{Basic technique}
\label{sec:uniformly}
The following lemma exemplifies proving a variant of a random recursive tree has logarithmic height.
The argument here is inspired by a proof in
Frieze and Tsourakakis~\cite{first}.
We will use a simple inequality:
let $a_1,\dots,a_m$ be positive numbers, and let $h\in[m]$.
Then observe that 
$$
\sum_{1 \le t_1 < \dots < t_h\le m}\ \left( \prod_{k=1}^{h}a_{t_k} \right)
< \frac{1}{h!}
\left(\sum_{i=1}^{m}a_i\right)^{h} \:.$$

\begin{lemma}
\label{lem:multichildren}
Let $(A_t)_{t\in\N}$ be a sequence of $\N$-valued  random variables.
Consider a growing tree $(T_t)_{t=0}^{\infty}$ as follows.
$T_0$ is arbitrary.
At each time-step $t\in\N$, 
a random vector $(W_1,W_2,\dots,W_{A_t}) \in V(T_{t-1})^{A_t}$
is chosen in such a way that for each $i\in[A_t]$
and each $v\in V(T_{t-1})$,
the marginal probability $\p{W_i=v}$ equals $|V(T_{t-1})|^{-1}$. 
In other words, each $W_i$ is a node of $T_{t-1}$ sampled uniformly;
however, the $W_j$'s may be correlated.
Then $A_t$ new nodes $v_1,\dots,v_{A_t}$ are born and $v_i$ is joined to $W_i$ for each $i\in[A_t]$.
Let $\ell=\ell(n),u=u(n)$ be positive integers such that $\ell\le A_t \le u$ for all $t\in[n]$.
Then the height of $T_n$ is a.a.s.\ at most $(u/\ell) e\log n + 2ue+O(1)$.
\end{lemma}

Note that we do not require any independence for $(A_t)_{t\in\N}$. In particular they can be correlated and depend on the past and the future of the process.

\begin{proof}
Let $n_0=|V(T_0)|$.
For a given integer $h = h(n)$, let us bound the probability that $T_n$ has a node at depth exactly $h+n_0$.
Given a sequence $1 \le t_1 < t_2<\dots < t_h\le n$,
the probability that there exists a path $v_{t_1}v_{t_2}\dots v_{t_h}$ in $T_n$ such that $v_{t_j}$ is born at time $t_j$ is at most
$$
u^{h} \prod_{k=2}^{h}\frac{1}{n_0+ \ell\cdot(t_k-1)} \:,
$$
since there are at most $u^{h}$ choices for $(v_{t_1},\dots,v_{t_h})$, 
and for each $k=2,3,\dots,h$, when $v_{t_k}$ is born, there are at least $n_0+ \ell\cdot(t_k-1)$ nodes available for it to join to.
By the union bound, the probability that $T_n$ has a node at depth $h+n_0$ is at most
\begin{align*}
u^{h}\sum_{1 \le t_1 < \dots < t_h\le n}\ \left( \prod_{k=2}^{h}\frac{1}{n_0+ \ell\cdot(t_k-1)}\right)
& < \frac{u^{h}}{h!}\left(1 + \sum_{j=1}^{n-1}\frac{1}{n_0+\ell j}\right)^{h} \\
& < \left(\frac{u e}{h} 
\cdot \left( \frac{\log n}{\ell} + 2\right) \right)^h
\bigg/\sqrt{2\pi h} \:,
\end{align*}
where we have used Stirling's formula
and the inequality $1 + \frac12 + \frac13 + \dots + \frac{1}{n-1} < 1 +\log n$.
Putting $h \ge (u/\ell) e\log n + 2ue$ makes this probability $o(1)$.
Hence a.a.s.\ the height of $T_n$ is at most 
$(u/\ell) e\log n + 2ue + n_0$, as required.
\end{proof}

Given $p\in(0,1]$, let $\geo(p)$ denote a geometric random variable with parameter $p$; namely 
$\p{\geo(p) = k} = (1-p)^kp$
for every $k\in\Nz$.
The first model we study is the \emph{basic forest fire model}
of Leskovec, Kleinberg, and Faloutsos~\cite[Section~4.2.1]{forest_fire_journal}.

\begin{model}
\label{def:forest_fire}
Let $p,q\in[0,1]$ be arbitrary.
We build a growing directed graph as follows.
$G_0$ is an arbitrary weakly connected directed graph.
At each time-step $t\in\N$, a new vertex $v$ is born and 
edges are created from it to the existing graph using the following process.
\begin{enumerate}
\item All vertices are marked `unvisited.' An \emph{ambassador vertex} $W$ is sampled uniformly from the existing graph.
\item Vertex $v$ is joined to $W$ and $W$ is marked as `visited.'
\item We independently generate two random variables $X=\geo(p)$ and $Y=\geo(q)$. 
We randomly select $X$ unvisited out-neighbours and $Y$ unvisited in-neighbours of $W$.
If not enough unvisited in-neighbours or out-neighbours are available, we select as many as we can.
Let $W_1,\dots,W_{Z}$ denote these vertices. 
\item Vertex $v$ is joined to $W_1, \dots, W_{  Z}$, then we apply steps 2--4 recursively to each of $W_1, \dots, W_{Z  }$.
\end{enumerate}
\end{model}

\begin{theorem}
Consider $(G_t)_{t=0}^{\infty}$ generated by Model~\ref{def:forest_fire}.
A.a.s.\ for every vertex $v$ of $G_n$ there exists a directed path of length at most $e\log n + O(1)$ connecting $v$ to some vertex of $G_0$.
In particular, a.a.s.\ the diameter of $G_n$ is at most $2e\log n + O(1)$.
\end{theorem}

\begin{proof}
We define a growing tree  $(T_t)_{t=0}^{\infty}$ in such a way that $T_t$ is a spanning tree of $G_t$ for all $t\in\Nz$:
$T_0$ is an arbitrary spanning tree of $G_0$.
For every $t\in\N$, if $v$ is the vertex born at time $t$ and $w$ is the corresponding ambassador vertex, then $v$ is joined only to $w$ in $T_t$.
By Lemma~\ref{lem:multichildren}, a.a.s.\ the height of $T_n$ is at most $e\log n + O(1)$.
\end{proof}

We next study the \emph{linear growth copying model} of Kumar, Raghavan, Rajagopalan, Sivakumar, Tomkins, and Upfal~\cite[Section~2.1]{copying_model_def}.

\begin{model}
\label{def:copying}
Let $p\in[0,1]$ and $d\in\N$.
We build a growing directed graph in which every vertex has out-degree $d$, and there is a fixed ordering of these $d$ edges.
$G_0$ is an arbitrary weakly connected directed graph with all vertices having out-degree $d$.
In each time-step $t\in\N$ a new vertex $v$ is born
and $d$ outgoing edges from $v$ to the existing graph are added, as described below.
An ambassador vertex $W$ is sampled uniformly from the existing vertices.
For $i\in[d]$, the head of the $i$-th outgoing edge of $v$ is chosen as follows:
with probability $p$, it is a random vertex of the existing graph sampled uniformly,
and with probability $1-p$ it is the head of the $i$-th outgoing edge of $W$, in which case we say $v$ has \emph{copied} the $i$-th outgoing edge of $W$.
\end{model}

\begin{theorem}
A.a.s.\ the diameter of $G_n$ defined in Model~\ref{def:copying} is at most $4e\log n + O(1)$.
\end{theorem}

\begin{proof}
We inductively define a growing tree  $(T_t)_{t=0}^{\infty}$ in such a way that 
the node set of $T_t$ equals the vertex set of $G_t$
for all $t$.
We prove by induction that for each $v\in V(G_t)$,
$\depth(v,G_t) \le 2 \depth(v,T_t)$.
Let $T_0$ be a breadth-first search tree of $G_0$, rooted at the root of $G_0$.
For each $t\in\N$, let $v$ be the vertex born at time $t$, and let $w$ be the corresponding ambassador vertex.
We consider two cases:
\begin{description}
\item[Case 1. $v$  copies at least one outgoing edge of $w$.]
In this case, we join $v$ to $w$ in $T_t$.
Since $v$ and $w$ have distance 2 in $G_t$, 
$\depth(v,G_t)\le \depth(w,G_t)+2$, so by the induction hypothesis for $w$,
$$
\depth(v,G_t) 
\le \depth(w,G_t) + 2
\le 2 \depth(w,T_t) + 2
= 2 \depth(v,T_t)
\:,$$
as required.

\item[Case 2. $v$ does not copy any outgoing edge of $w$.]
Let $x$ denote the head of the first outgoing edge of $v$.
In this case, we join $v$ to $x$ in $T_t$.
Using the induction hypothesis for $x$,
\begin{align*}
\depth(v,G_t) 
\le \depth(x,G_t) + 1
& \le 2 \depth(x,T_t) + 1 \\
& < 2 \depth(x,T_t) + 2
= 2 \depth(v,T_t)
\:,
\end{align*}
as required.
\end{description}
Notice that in either case, node $v$ is joined to a node of $T_{t-1}$ sampled uniformly.
By Lemma~\ref{lem:multichildren}, a.a.s.\ the height of $T_n$ is at most $e\log n + O(1)$,
so a.a.s.\ the diameter of $G_n$ is at most $4e \log n + O(1)$.
\end{proof}


\subsection{Sampling neighbours using PageRank}
\label{sec:pagerank}
In this section we study a model in which the neighbours of each new vertex are chosen according to the PageRank distribution.
We recall the definition of PageRank.

\begin{definition}[PageRank~\cite{pagerank_def}]
\label{def:pagerank}
Let $q\in[0,1]$ and let $G$ be a directed graph.
\emph{PageRank} is the unique probability distribution $\pi_q : V(G) \to [0,1]$ that satisfies
\begin{equation}
\label{eq:pagerank}
\pi_q (v) = \frac{1-q}{|V(G)|} + q \sum_{u\in V(G)} \frac{\pi_q(u) \cdot \#(uv)}{\outdeg(u)} \:.
\end{equation}
Here $\#(uv)$ denotes the number of copies of the directed edge $uv$ in the graph (which is zero if there is no edge from $u$ to $v$), and
$\outdeg(u)$ denotes the out-degree of $u$.
\end{definition}
PageRank is used as a ranking mechanism in Google~\cite{google}.
More details and applications can be found in~\cite{pagerank_deep}.

\begin{model}
\label{def:pagerankmodel}
Let $p_a,p_b,p_c$ be nonnegative numbers summing to 1, let $q\in[0,1]$ and $d\in\N$.
We build a growing directed graph $(G_t)_{t=0}^{\infty}$ in which every vertex has out-degree $d$.
$G_0$ is a weakly connected directed graph with all vertices having out-degree $d$.
In each time-step $t\in\N$, a new vertex is born
and $d$ outgoing edges from it to the existing graph are added.
The heads of the new edges are chosen independently.
For choosing the head of each edge, we perform one of the following operations, independently of previous choices.
\begin{enumerate}[(a)]
\item
With probability $p_a$, the head is a vertex sampled uniformly from the existing graph.
\item
With probability $p_b$, it is the head of an edge sampled uniformly from the existing graph.
\item
With probability $p_c$, it is a vertex sampled from the existing graph using $\pi_q$.
\end{enumerate}
\end{model}

Model~\ref{def:pagerankmodel} is defined by Pandurangan, Raghavan, and Upfal~\cite[Section~2]{pagerank_model_journal}.
They call it the \emph{hybrid selection model}.
For the special case $p_b=0$,
which is referred to as the \emph{PageRank-based selection model}, it has been proved using a different argument that a.a.s.\ the diameter is $O(\log n)$~\cite{we_surfer}.
For bounding the diameter of Model~\ref{def:pagerankmodel} we will need a lemma.

\begin{lemma}
\label{lem:cm}
Assume that $q<1$.
There exists a random variable $L$ such that the head of each new edge in Model~\ref{def:pagerankmodel} can be obtained by sampling a vertex $W$ uniformly from the existing graph 
and performing a simple random walk of length $L$ starting from $W$.
Moreover, $L$ is stochastically smaller than $1+\geo(1-q)$.
\end{lemma}

\begin{proof}
We claim that
$$
L = \begin{cases}
0 & \mathrm{with\ probability\ } p_a \:,\\
1 & \mathrm{with\ probability\ } p_b \:,\\
\geo(1-q) & \mathrm{with\ probability\ } p_c \:. 
\end{cases}
$$
If we sample a vertex uniformly and perform a random walk of length 1, then since all vertices have the same out-degree, the last vertex of the walk is the head of a uniformly sampled edge.

So it suffices to show that if we sample a vertex uniformly and perform a random walk of length $\geo(1-q)$, the last vertex of the walk has distribution $\pi_{q}$.
This was first observed in~\cite{pagerank_random_surfer}.
Let ${\tau}\in[0,1]^{V(G)}$ denote the probability distribution of the last vertex, 
let $\mathcal{P}$ denote the probability transition matrix of the simple random walk, and let $\sigma = \big[1/|V(G)|,1/|V(G)|,\dots,1/|V(G)|\big]^T$
be the uniform distribution.
Then we have
$$
\tau 
= \sum_{k=0}^{\infty} q^k (1-q) \mathcal{P}^k \sigma
= (1-q) \sigma + 
q\mathcal{P}
\left(\sum_{k=1}^{\infty} q^{k-1} (1-q) \mathcal{P}^{k-1} \sigma\right)
=
(1-q) \sigma
+
q\mathcal{P} \tau \:.
$$
Comparing with (\ref{eq:pagerank}) and noting that the stationary distribution of an Ergodic Markov chain is unique, we find that $\tau=\pi_{q}$, as required.
\end{proof}

\begin{paragraph}{The Chernoff bound}
Given $n\in\Nz$ and $p\in[0,1]$,
let $\bin(n,p)$ denote a binomial random variable with parameters $n$ and $p$.
We refer to the following inequality, valid for every $\varepsilon\ge0$, as the \emph{Chernoff bound}.
See Motwani and Raghavan~\cite[Theorem 4.2]{rand_algs} for a proof.
$$\p{\bin(n,p)<(1-\varepsilon)np} \le \exp(-\varepsilon^2 np / 2)\:.$$
\end{paragraph}
\begin{theorem}
If $q<1$, then a.a.s.\ the diameter of $G_n$ defined in Model~\ref{def:pagerankmodel} is at most $18 \log n /(1-q)$.
\end{theorem}

\begin{proof}
A \emph{weighted tree} is a tree with nonnegative weights  assigned to the edges.
The \emph{weighted depth} of a node $v$ is defined as the sum of the weights of the edges connecting $v$ to the root.
We define a growing weighted tree  $(T_t)_{t=0}^{\infty}$ 
such  that for all $t$,
the node set of $T_t$ equals the vertex set of $G_t$.
We prove by induction that the depth of each vertex in $G_t$ is at most its weighted depth in $T_t$.
Let $T_0$ be a breadth-first search tree of $G_0$ rooted at the root of $G_0$, 
and let all edges of $T_0$ have unit weights.
Assume that when obtaining $G_t$ from $G_{t-1}$, the heads of the new edges are chosen using the procedure described in Lemma~\ref{lem:cm}.
For every $t\in\N$, if $v$ is the vertex born at time $t$, and $w$ and $l$ are the first sampled vertex and length of the first random walk taken, respectively, then $v$ is joined only to $w$ in $T_t$ and the weight of the edge $vw$ is set to $l+1$.
Note that the edge weights are mutually independent.
Since the distance between $v$ and $w$ in $G_t$ is at most $l+1$,
by induction the weighted depth of $v$ in $T_t$ is at most the depth of $v$ in $G_t$.
We show that a.a.s.\ the weighted height of $T_n$ is at most $ 9  \cdot \log n /(1-q)$, and this completes the proof.

By Lemma~\ref{lem:multichildren}, a.a.s.\ the (unweighted) height of $T_n$ is less than $ 1.001e\log n$.
We prove that any given node at depth at most $ 1.001e\log n$ of $T_n$  has weighted depth at most $  9(\log n)/(1-q)$ with probability $1-o(1/n)$, and then the union bound completes the proof.
Let $v$ be a node of $T_n$ at depth $h$, where $h \le 1.001 e \log n$.
By Lemma~\ref{lem:cm}, the weighted depth of $v$ equals the sum of $h$ independent random variables, each stochastically smaller than $2+\geo(1-q)$.
The probability that the sum of $h$
independent random variables distributed as $2+\geo(1-p)$
is greater than $9\log n / (1-q)$
is
$$ \p{\bin \left(\frac{9\log n}{1-q} - h , 1-q\right) < h}\:.$$
Since $h \le 1.001 e \log n$, using the Chernoff bound we infer that this probability is less than
$
\exp\left(-0.566^2 \times 6.27 (\log n)/2\right)
< n^{-1.004}
$.
\end{proof}

\section{Incorporating preferential attachment: edge trees}
\label{sec:pref}

In this section we study models incorporating preferential attachment.
We first define a model that has a lot of flexibility (Model~\ref{def_pref}) and prove it has logarithmic diameter.
Then we reduce Models~\ref{def_acld},~\ref{def_butowsley}, and~\ref{def_aclc}
to this model.

\begin{model}
\label{def_pref}
Let $(A_t,B_t)_{t=1}^{\infty}$ be sequences of $\Nz$-valued random variables.
Consider a growing undirected graph $(G_t)_{t=0}^{\infty}$ as follows.
$G_0$ is an arbitrary connected graph with at least one edge.
At each time-step $t\in\N$, $G_t$ is obtained from $G_{t-1}$ by doing a vertex operation and an edge operation, as defined below.

In a \emph{vertex operation}, if $A_t>0$, a new vertex is born and $A_t$ edges are added in the following manner:
we sample an edge uniformly from $G_{t-1}$, choose one of its endpoints arbitrarily, and join it to the new vertex.
For the other $A_t - 1$ new edges, one endpoint is the new vertex, and the other endpoint is arbitrary (can be the new vertex as well).
If $A_t=0$ then the vertex operation does nothing.

In an \emph{edge operation}, we independently sample $B_t$ edges uniformly from $G_{t-1}$ and we choose an arbitrary endpoint of each sampled edge.
Then we add $B_t$ new edges, joining these vertices to arbitrary vertices of $G_{t-1}$.
\end{model}

Note that we do not require any independence for $(A_t,B_t)_{t\in\N}$. In particular they can be correlated and depend on the past and the future of the process.

A novel idea in this paper is introducing edge trees: these are
trees coupled with graphs whose nodes correspond to the edges of the graph.
The following theorem demonstrates their usage.

\begin{theorem}
\label{thm:pref}
Let $\ell=\ell(n),u=u(n) \in \N$ be such that $\ell\le A_t + B_t \le u$ for every $t\in\N$.
A.a.s.\ the graph $G_n$ generated by Model~\ref{def_pref} has diameter at most $ 4e(u/\ell) \log n + 8eu + O(1)$.
\end{theorem}

\begin{proof}
We define the \emph{depth} of an edge $xy$ as $1+\min\{\depth(x),\depth(y)\}$.
We inductively define a growing tree $(T_t)_{t=0}^{\infty}$ such that for all $t\in\Nz$,
$V(T_t) = E(G_t) \cup \{\root\}$.
Here $\root$ denotes the root of $T_t$, which has depth 0.
We prove by induction that for all $e\in E(G_t)$,
$\depth(e, G_t) \le 2 \depth(e,T_t)$.
Let $H$ be the graph obtained from $G_0$ by adding an edge labelled $\root$ incident to its root.
Let $T_0$ be a breadth-first tree of the line graph of $H$ rooted at $\root$. 
(The \emph{line graph} of a graph $H$ is a graph whose vertices are the edges of $H$, and two edges are adjacent if they have a common endpoint.)
Note that $\depth(\root,T_0)=0$ and 
$\depth(e,T_0)=\depth(e,G_0)$ for every $e\in E(G_0)$.

Given $T_{t-1}$, we define $T_t$ and prove the inductive step.
First, consider a vertex operation with $A_t>0$.
Let $v$ be the new vertex, $e_1$ be the sampled edge, and $w_1$ be the chosen endpoint of $e_1$.
Notice that $\depth(v,G_t) \le \depth(e_1,G_t)+1$.
In $T_t$, we join the $A_t$ edges incident with $v$ to $e_1$.
For any such edge $e$ we have
\begin{align*}
\depth(e,G_t)
\le \depth(v,G_t) + 1
\le \depth(e_1,G_t) + 2
\le 2 \depth(e_1,T_t) + 2
= 2 \depth(e,T_t),
\end{align*}
where we have used the inductive hypothesis for $e_1$ in the third inequality.

Second, consider an edge operation.
Let $e_1,\dots,e_{B_t}$ be the sampled edges, and let $w_1,w_2,\dots,w_{B_t}$ be the chosen endpoints.
For each $j\in[B_t]$, in $G_t$ we join $w_j$ to some vertex of $G_{t-1}$, say $x_j$.
In $T_t$, we join the new edge $w_j x_j$ to $e_j$.
We have
\begin{align*}
\depth(w_j x_j,G_t)
\le \depth(w_j,G_t) + 1
& \le \depth(e_j,G_t) + 1\\
& \le 2 \depth(e_j,T_t) + 1
= 2 \depth(w_j x_j,T_t) - 1 \:,
\end{align*}
where we have used the fact that $e_j$ is incident to $w_j$ in the second inequality, and the 
inductive hypothesis for $e_j$ in 
the third inequality.

Hence for all $e\in E(G_t)$, $\depth(e, G_t) \le 2 \depth(e,T_t)$. On the other hand, examining the construction of $(T_t)_{t\in\Nz}$ and using Lemma~\ref{lem:multichildren}, we find that a.a.s.\ the height of $T_n$ is at most $ (u/\ell) e\log n + 2ue+O(1)$.
This implies that a.a.s.\ the diameter of $G_n$ is at most $4(u/\ell) e\log n + 8ue+O(1)$.
\end{proof}

\begin{definition}[$\rho$]
For an undirected graph $G$ and a real number $\delta$, we define the function $\rho_{\delta}:V(G)\to\mathbb{R}$ as 
$$\rho_{\delta}(v) = \frac{\deg(v)+\delta}{\sum_{u\in V(G)} (\deg(u)+\delta)} \:.$$
Here $\deg(v)$ denotes the degree of vertex $v$,
and a loop is counted twice.
Note that if $\delta>-1$ then $\rho_{\delta}$ is a probability distribution.
\end{definition}

Observe that to sample a vertex using $\rho_0$,
one can sample an edge uniformly and then choose one of its endpoints uniformly.
Most of our arguments are based on this crucial fact,
and this is the reason for introducing edge trees.

\begin{model}
\label{def_acld}
Let $\{X_t: t\in\N\}$ be a sequence of $\N$-valued random variables,
and let $\{Y_t,Z_t : t\in\N\}$ be sequences of $\Nz$-valued random variables.
We consider a growing undirected graph $(G_t)_{t=0}^{\infty}$ as follows.
$G_0$ is an arbitrary connected graph with at least one edge.
At each time-step $t\in\N$, $G_t$ is obtained from $G_{t-1}$ by performing the following three operations.
\begin{enumerate}
\item
We sample $X_t$ vertices $N_1,\dots,N_{X_t}$ independently using $\rho_0$.
\item
We sample $2Z_t$ vertices $W_1,W'_1,W_2,W'_2,\dots,W_{Z_t},W'_{Z_t}$ independently using $\rho_0$.
\item
We add a new vertex $v$ and add the edges $W_1W'_1,\dots,W_{Z_t}W'_{Z_t}$, $vN_1,\dots,vN_{X_t}$.
We also add $Y_t$ loops at $v$.
\end{enumerate}
\end{model}

Model~\ref{def_acld} is a generalization of a model defined by Aiello, Chung, and Lu~\cite[Section~2.1, Model D]{acl_model}, which has bounded $X_t,Y_t,Z_t$.
The following theorem implies that a.a.s.\ the latter model has diameter $O(\log n)$.

\begin{theorem}
Let $\ell=\ell(n),u=u(n)$ be positive integers such that $X_t>0$ and $\ell \le X_t+Y_t+Z_t\le u$ for all $t\in\N$.
A.a.s.\ the diameter of $G_n$ generated by Model~\ref{def_acld} is at most $4e(u/\ell) \log n + 8eu+O(1)$.
\end{theorem}

\begin{proof}
We claim that $(G_t)_{t=0}^{\infty}$ grows as described in Model~\ref{def_pref}.
Sampling a vertex using $\rho_{0}$ corresponds to choosing a random endpoint of a random edge.
The three operations of Model~\ref{def_acld} correspond to
applying a vertex operation with $A_{t} = X_t+Y_t$
and an edge operation with $B_t=Z_t$.
By Theorem~\ref{thm:pref}, a.a.s.\ the diameter of ${G}_{n}$ is at most $4e(u/\ell) \log n + 8eu+O(1)$.
\end{proof}

We analyse another model by reducing it to Model~\ref{def_pref}.

\begin{model}
\label{def_butowsley}
Let $\delta \in (-1,\infty)$, $p\in[0,1]$ and let 
$(X_t)_{t\in\N}$ be a sequence of $\N$-valued random variables.
We consider a growing undirected graph $(G_t)_{t=0}^{\infty}$ as follows.
$G_0$ is an arbitrary connected graph with at least one edge.
At each time-step $t\in\N$, we apply 
exactly one of the following operations: operation (a) with probability $p$ and operation (b) with probability $1-p$.
\begin{enumerate}[(a)]
\item We sample $X_t$ vertices independently using $\rho_{\delta}$, then we add a new vertex $v$ and join it to the sampled vertices.
\item We sample $2X_t$ vertices $W_1,W'_1,\dots,W_{X_t},W'_{X_t}$ independently using $\rho_{\delta}$,
then we add the edges $W_1W'_1, \dots,W_{X_t}W'_{X_t}$.
\end{enumerate}
\end{model}

Model~\ref{def_butowsley} is a generalization of the \emph{generalized linear preference model} of Bu and Towsley~\cite{bu_towsley}, in which $X_t=d$ for all $t$, where $d$ is a fixed positive integer.
Theorem~\ref{thm:remark} below gives that if $\delta$ is rational and nonnegative then a.a.s.\ the generalized linear preference model has diameter at most $(4+2\delta/d)e\log n + O(1)$.

\begin{remark}
Model~\ref{def_butowsley} with $p=1$ and $X_t$ being a constant independent of $t$ and $n$ is called the \emph{linear preference model}, whose diameter has been studied extensively.
Assume that $X_t=d$ for all $t$, where $d\in\N$ is fixed.
If $d=1$ and $\delta\ge0$, Pittel~\cite{random_recursive_trees} showed the diameter is $\Theta(\log n)$.
If $d>1$ and $\delta\in(-d,0)$,  the diameter is
$\Theta(\log \log n)$ as proved by
Dommers, van der Hofstad, and Hooghiemstra~\cite{diameters_pa,vander}.
If $d>1$ and $\delta=0$,  the diameter is
$\Theta (\log n / \log \log n)$, see Bollob{\'a}s and Riordan~\cite{diameter_preferential_attachment}.
Finally, if $d>1$ and $\delta>0$,  the diameter is $\Theta(\log n)$~\cite{diameters_pa,vander}.

Chung and Lu~\cite{coupling_online_offline} studied a variation of Model~\ref{def_butowsley} with the following differences:
the process is conditioned on generating a graph with no multiple edges or loops;
$X_t=d$ for all $t$, where $d$ may depend on $n$;
there are two additional operations: in the first one, a vertex is sampled uniformly and deleted, 
and in the second one, $X_t$ edges are sampled uniformly and deleted.
They proved that if $d > \log ^{1+\Omega(1)} n$,
then a.a.s.\ the evolving graph has diameter $\Theta(\log n)$,
where $n$ is the number of vertices.
\end{remark}

\begin{theorem}
\label{thm:remark}
Suppose that $\delta=r/s$, where $r\in\Nz$ and $s\in\N$,
and suppose that $\ell=\ell(n),u=u(n)\in\N$ are such that 
$\ell \le X_t\le u$ for all $t$.
A.a.s.\ the diameter of $G_n$ generated by Model~\ref{def_butowsley} is at most 
$4e(u/\ell + \delta/(2\ell)) \log n + O(u)$.
\end{theorem}

\begin{proof}
For $t\in\Nz$, let $\widehat{G}_{t}$ be the graph obtained from $G_t$ by copying each edge $2s-1$ times,
and adding $r$ loops at each vertex.
So $\widehat{G}_{t}$ has $2s |E(G_t)| + r |V(G_t)|$ edges.
Note that the diameters of $G_t$ and $\widehat{G}_{t}$ are the same.
We claim that $(\widehat{G}_t)_{t=0}^{\infty}$ grows as described in Model~\ref{def_pref}.
First, sampling a vertex of $G_{t-1}$ using $\rho_{\delta}$ corresponds to choosing a random endpoint of a random edge of $\widehat{G}_{t-1}$.
Second, applying operation (a) corresponds to applying only a vertex operation with $A_{t}=2sX_t+r$.
Second, applying operation (b) corresponds to applying only an edge operation with $B_{t} = 2sX_t$.
By Theorem~\ref{thm:pref}, a.a.s\ the diameter of $\widehat{G}_{n}$ is at most $4e(u/\ell + \delta/(2\ell)) \log n + (16esu + 8er)+O(1)$, completing the proof.
\end{proof}

We now analyse the preferential attachment with random initial degrees (PARID) model of Deijfen, van den Esker, van der Hofstad, and Hooghiemstra~\cite[Section~1.1]{randominitial} by reducing it to Model~\ref{def_butowsley}.

\begin{model} 
\label{def_parid}
Let $\{X_t: t\in\N\}$ be a sequence of i.i.d.\ $\N$-valued random variables and let $\delta$ be a fixed number such that almost surely $X_1 + \delta > 0$.
We consider a growing undirected graph $(G_t)_{t=0}^{\infty}$ as follows.
$G_0$ is an arbitrary connected graph with at least one edge.
At each time-step $t\in\N$, $G_t$ is obtained from $G_{t-1}$ by sampling
$X_t$ vertices $N_1,\dots,N_{X_t}$ independently using $\rho_{\delta}$
and adding one new vertex $v$ and $X_t$ new edges $vN_1,\dots,vN_{X_t}$.
\end{model}

Note that $X_t$ is the (random) initial degree of the vertex born at time $t$.
The following theorem implies that if the initial degrees' distribution in the PARID model has an exponential decay (e.g.\ if it is the Poisson or the geometric distribution),
and $\delta$ is positive and rational,
then a.a.s.\ the generated graph has a polylogarithmic diameter.

\begin{theorem}
\label{thm:parid}
Assume that $\delta$ is a positive rational number
and that $\ell=\ell(n)$ and $u=u(n)$ are positive integers such that 
$\p{X_1 \notin[\ell,  u]} = o(1/n)$.
A.a.s.\ the diameter of $G_n$ generated by Model~\ref{def_parid} is at most 
$4e(u/\ell + \delta/(2\ell)) \log n + O(u)$.
\end{theorem}

\begin{proof}
Since
$\p{X_1 \notin[\ell,  u]} = o(1/n)$
and the $X_i$ are i.i.d.,
a.a.s.\ we have $\ell \le X_t \le u$ for all
$t\in[n]$.
The rest of the proof is the same as that of Theorem~\ref{thm:remark}, where all operations are of type (a).
\end{proof}

\subsection{A directed model}

In this section we study a directed analogous of Model~\ref{def_acld},
which is also a generalization of a model of Aiello et al.~\cite{acl_model}.
Sampling probabilities in this model depend on vertices' out-degrees and in-degrees, as defined below.

\begin{definition}[$\rho^{out},\rho^{in}$]
For a directed graph $G$ and a real number $\delta$, we define the functions $\rho_{\delta}^{out},
\rho_{\delta}^{in}:V(G)\to\mathbb{R}$ as 
$$\rho^{out}_{\delta}(v) = \frac{\outdeg(v)+\delta}{\sum_{u\in V(G)} (\outdeg(u)+\delta)}$$ and
$$\rho^{in}_{\delta}(v) = \frac{\indeg(v)+\delta}{\sum_{u\in V(G)} (\indeg(u)+\delta)} \:.$$
Here $\outdeg(v)$ and $\indeg(v)$ denote the out-degree and the in-degree of vertex $v$, respectively.
\end{definition}

\begin{model}
\label{def_aclc}
Let $\{X_t,Y_t,Z_t,Q_t : t\in\N\}$ be sequences of $\Nz$-valued random variables
satisfying $X_t+Y_t>0$ for all $t$.
We consider a growing directed graph $(G_t)_{t=0}^{\infty}$ as follows.
$G_0$ is an arbitrary weakly connected directed graph with at least one edge.
At each time-step $t\in\N$, we perform the following operations:
\begin{enumerate}
\item
We sample $X_t$ vertices $x_1,\dots,x_{X_t}$ independently using $\rho_0^{out}$
and $Y_t$ vertices $y_1,\dots,y_{Y_t}$ independently using $\rho_0^{ in}$.
\item
We sample $Z_t$ vertices $w_1,w_2,\dots,w_{Z_t}$ independently using $\rho_0^{out}$, and we sample 
$Z_t$ vertices $w'_1,w'_2,\dots,w'_{Z_t}$ independently using $\rho_0^{in}$,
\item
We add a new vertex $v$, and then we add the directed edges 
$w_1w'_1,\dots,w_{Z_t}w'_{Z_t}$,
$x_1v,\dots,x_{X_t}v$,
$vY_1,\dots,vY_{Y_t}$.
We also add $Q_t$ loops at $v$.
\end{enumerate}
\end{model}

Model~\ref{def_aclc} generalizes of~\cite[Section~2.1, Model C]{acl_model},
which has bounded $X_t,Y_t,Z_t,Q_t$.
The following theorem implies that a.a.s.\ the diameter of the latter model is $O(\log n)$.

\begin{theorem}
Let $\ell=\ell(n),u=u(n)$ be positive integers such that  $\ell \le X_t+Y_t+Z_t+Q_t\le u$ for all $t\in\N$.
A.a.s.\ the diameter of $G_n$ generated by Model~\ref{def_aclc} is at most $4 e(u/\ell) \log n + 8eu+O(1)$.
\end{theorem}

\begin{proof}
We claim that the underlying undirected graph of $(G_t)_{t=0}^{\infty}$ grows as described in Model~\ref{def_pref}.
Sampling a vertex using $\rho_{0}^{out}$ and $\rho_{0}^{in}$ correspond to choosing the tail and the head of a random edge, respectively.
The operations of Model~\ref{def_acld} correspond to
applying a vertex operation with
$A_t = X_t+Y_t+Q_t$
and an edge operation with $E_t = Z_t$.
By Theorem~\ref{thm:pref}, a.a.s.\ the diameter of $G_{n}$ is at most $4 e(u/\ell) \log n + 8eu+O(1)$, as required.
\end{proof}

\section{Directed scale-free graphs: dummy edges}
\label{sec:directed}
We study two directed models in this section.
In contrast to the previous directed model (Model~\ref{def_aclc}), 
in models considered here, the constant term in the definition of attachment probabilities ($\delta$ in Model~\ref{def_butowsley}) can be different for in-degrees and out-degrees.
We handle this issue by introducing dummy edges whose role is just to adjust the attachment probabilities 
(similar to, but more complicated than, what we did in the proof of Theorem~\ref{thm:remark}).
As in Section~\ref{sec:pref}, we first define a general model (Model~\ref{def_directed}) with a lot of flexibility and prove that a.a.s.\ it has a logarithmic diameter, and then reduce Model~\ref{def_directed_scale_free} (which is a generalization of the so-called `directed scale-free graphs') to that.

\begin{definition}[generalized directed graph]
In a directed graph, each edge has a tail and a head. 
A \emph{generalized directed graph} is a directed graph some of whose edges do not have a head or a tail.
Edges of such a graph are of three type:
\emph{tailless} edges have a head but do not have a tail,
\emph{headless} edges have a tail but do not have a head,
and \emph{proper} edges have a tail and a head.
A \emph{headed} edge is one that is not headless,
and a \emph{tailed} edges is one that is not tailless.
\end{definition}

The following model is a directed analogous of Model~\ref{def_pref}.

\begin{model}
\label{def_directed}
Let $(A_t,B_t,C_t,D_t,E_t)_{t=1}^{\infty}$ be sequences of $\Nz$-valued random variables.
We consider a growing generalized directed graph $(G_t)_{t=0}^{\infty}$ as follows.
$G_0$ is an arbitrary weakly connected generalized directed graph with at least one edge.
At each time-step $t\in\N$, $G_t$ is obtained from $G_{t-1}$ by performing a vertex operation and an edge operation, as defined below.

In a \emph{vertex operation}, if $A_t+B_t>0$, a new vertex $v$ is born and $A_t+B_t+C_t+D_t$ edges are added in the following manner:
\begin{description}
\item[Case 1:] If $A_t>0$, we sample a headed edge from $G_{t-1}$ uniformly and add a proper edge from $v$ to its head.
Then we add $A_t - 1$ new proper edges, tailed at $v$ and headed at arbitrary vertices of $G_{t-1}$.
Then $B_t$ proper edges are added, tailed at arbitrary vertices of $G_{t-1}$ and headed at $v$.
Then $C_t$ headless edges tailed at $v$, and $D_t$ tailless edges headed at $v$ are added.

\item[Case 2:] If $A_t=0$, we sample a tailed edge from $G_{t-1}$ uniformly and add a proper edge from its tail to $v$.
Then $B_t-1$ new proper edges are added from arbitrary vertices of $G_{t-1}$ to $v$.
Then $C_t$ headless edges tailed at $v$, and $D_t$ tailless edges headed at $v$ are added.
\end{description}
If $A_t+B_t=0$, then we do nothing in the vertex operation.

In an \emph{edge operation}, we independently sample $E_t$ tailed edges from $G_{t-1}$ uniformly, then we add $E_t$ proper edges, joining the tails of the sampled edges to arbitrary vertices of $G_{t-1}$.
\end{model}

Note that we do not require any independence for $(A_t,B_t,C_t,D_t,E_t)_{t\in\N}$. In particular they can be correlated and can depend on the past and the future of the process.

\begin{theorem}
\label{thm:directed}
Let $\ell=\ell(n),u=u(n) \in \N$ be such that 
$\ell\le A_t + B_t+E_t$ and 
$A_t + B_t +C_t+D_t+E_t\le u$
for every $t\in\N$.
A.a.s.\ the graph $G_n$ generated by Model~\ref{def_directed} has diameter at most $4e(u/\ell) \log n + 8eu + O(1)$.
\end{theorem}

\begin{proof}
The argument is similar to that of Theorem~\ref{thm:pref}.
We define the depth of a headless edge as one plus the depth of its tail,
and
the depth of a tailless edge as one plus the depth of its head,
and the {depth} of a proper edge $uv$  as $1+\min\{\depth(u),\depth(v)\}$.
We inductively define a growing undirected tree $(T_t)_{t=0}^{\infty}$ such that for all $t\in\Nz$,
$V(T_t) = E(G_t) \cup \{\root\}$.
Here $\root$ denotes the root of $T_t$, which has depth 0.
We prove by induction that for all $e\in E(G_t)$,
$\depth(e, G_t) \le 2 \depth(e,T_t)$.
Let $H$ be the graph obtained from the underlying undirected graph of $G_0$ by adding an edge labelled $\root$ incident to its root.
Let $T_0$ be a breadth-first tree of the line graph of $H$ rooted at $\root$. 
Note that $\depth(\root,T_0)=0$ and 
$\depth(e,T_0)=\depth(e,G_0)$ for every $e\in E(G_0)$.

Given $T_{t-1}$, we define $T_t$ and prove the inductive step.
First, consider a vertex operation, Case 1.
Let $v$ be the new vertex and $e_1$ be the sampled headed edge.
Notice that $\depth(v,G_t) \le \depth(e_1,G_t)+1$.
In $T_t$, we join the $A_t+B_t+C_t+D_t$ new nodes (new edges of $G_t$) to $e_1$.
For any such edge $e$ we have
\begin{align*}
\depth(e,G_t)
\le \depth(v,G_t) + 1
\le \depth(e_1,G_t) + 2
\le 2 \depth(e_1,T_t) + 2
= 2 \depth(e,T_t),
\end{align*}
where we have used the inductive hypothesis for $e_1$ in the third inequality.

Second, consider a vertex operation, Case 2.
Let $v$ be the new vertex and let $e_1$ be the sampled tailed edge.
Notice that $\depth(v,G_t) \le \depth(e_1,G_t)+1$.
In $T_t$, we join the $B_t+C_t+D_t$ new nodes (new edges of $G_t$) to $e_1$.
For any such edge $e$ we have
\begin{align*}
\depth(e,G_t)
\le \depth(v,G_t) + 1
\le \depth(e_1,G_t) + 2
\le 2 \depth(e_1,T_t) + 2
= 2 \depth(e,T_t).
\end{align*}

Third, consider an edge operation.
Let $e_1,\dots,e_{E_t}$ be the sampled tailed edges, and denote by $w_1,w_2,\dots,w_{E_t}$ their tails.
For each $j\in[E_t]$, in $G_t$ we join $w_j$ to a vertex of $G_{t-1}$, say $x_j$.
In $T_t$, we join the new node $w_j x_j$ to $e_j$.
We have
\begin{align*}
\depth(w_j x_j,G_t)
\le \depth(w_j,G_t) + 1
& \le \depth(e_j,G_t) + 1 \\
& \le 2 \depth(e_j,T_t) + 1
= 2 \depth(w_j x_j,T_t) - 1 \:,
\end{align*}
where we have used the fact that $w_j$ is incident with $e_j$ for the second inequality, and the 
inductive hypothesis for $e_j$ in the third inequality.
Hence for all $e\in E(G_t)$, we have $\depth(e, G_t) \le 2 \depth(e,T_t)$, as required.
To complete the proof, it suffices to show that a.a.s.\ the height of $T_n$ is at most $(u/\ell) e\log n + 2ue+O(1)$.

The argument is similar to that for Lemma~\ref{lem:multichildren}.
Note that at any time $t$, graph $G_t$ has at least $|V(T_0)| +\ell  t$ proper edges.
Let $n_0=|V(T_0)|$.
For a given $h = h(n)$, we bound the probability that $T_n$ has a node at depth exactly $h+n_0$.
Given a sequence $1 \le t_1 < \dots < t_h\le n$,
the probability that there exists a path $v_{t_1}v_{t_2}\dots v_{t_h}$ in $T_n$ with $v_{t_j}$ born at time $t_j$ is at most
$$
u^h \prod_{k=2}^{h}\frac{1}{n_0+ \ell \cdot (t_k-1)} \:,
$$
since there are at most $u^h$ choices for 
$(v_{t_1},\dots,v_{t_h})$, and
for each $k=2,\dots,h$, when $v_{t_k}$ is born, there are at least $n_0+ \ell \cdot (t_k-1)$ nodes available for it to join to (corresponding to the proper edges of $G_{t_k-1}$).
By the union bound, the probability that $G_n$ has a node at depth $h+n_0$ is at most
\begin{align*}
u^h \sum_{1 \le t_1 <t_2< \dots < t_h\le n}\left( \prod_{k=2}^{h}\frac{1}{n_0+\ell\cdot(t_k-1)}\right)
& < \frac{u^h}{h!}\left(1+\sum_{j=1}^{n-1}\frac{1}{n_0+\ell j}\right)^h \\
& 
< \left(\frac{u e}{h} 
\cdot \left( \frac{\log n}{\ell} + 2\right) \right)^h
\bigg/\sqrt{2\pi h} \:.
\end{align*}
Putting $h \ge (u/\ell) e\log n + 2ue$  makes this probability $o(1)$.
Hence a.a.s.\ the height of $T_n$ is less than 
$ (u/\ell) e\log n + 2ue+O(1)$, as required.
\end{proof}

The following model is a directed analogous of Model~\ref{def_butowsley}.

\begin{model}
\label{def_directed_scale_free}
Let $p_a,p_b,p_c$ be nonnegative numbers summing to 1,
and let $\alpha,\beta \in [0,\infty)$.
Let
$(X_t)_{t\in\N}$ be a sequence of $\N$-valued random variables.
We consider a growing directed graph $(G_t)_{t=0}^{\infty}$ as follows.
$G_0$ is an arbitrary weakly connected directed graph.
At each time-step $t\in\N$, we perform exactly one of the following three operations,
with probabilities $p_a,p_b$, and $p_c$, respectively.
\begin{enumerate}
\item[(a)]
We sample $X_t$ vertices from the existing graph, independently using $\rho_{\alpha}^{in}$.
Then we add a new vertex and join it to the sampled vertices.
\item[(b)]
We sample $X_t$ vertices from the existing graph, independently using $\rho_{\beta}^{out}$.
Then we add a new vertex and join the sampled vertices to it.
\item[(c)]
We sample $X_t$ vertices
$w_1,\dots,w_{X_t}$ independently using $\rho_{\beta}^{out}$,
and we sample $X_t$ vertices
$w'_1,\dots,w'_{X_t}$ independently using $\rho_{\alpha}^{in}$.
Then we add the edges 
$w_1w'_1$, $\dots$, $w_{X_t}w'_{X_t}$.
\end{enumerate}
\end{model}

Model~\ref{def_directed_scale_free} is a generalization of \emph{directed scale-free graphs}
of Bollob{\'a}s, Borgs, Chayes, and Riordan~\cite[Section~2]{directed_scale_free}, which has $X_t=1$ for all $t$.
The following theorem implies that if $\alpha$ and $\beta$ are rational, then a.a.s.\ the diameter of the latter model is at most 
$4 e (1+\alpha+\beta)  \log n + O(1)$.

\begin{theorem}
Suppose that $\alpha=r/s$ and $\beta=q/s$ with $r,q\in\Nz$ and $s\in\N$.
Also suppose that $\ell=\ell(n),u=u(n)\in\N$ are such that $\ell \le X_t\le u$ for all $t$.
A.a.s.\ the diameter of $G_n$ generated by Model~\ref{def_directed_scale_free} is at most 
$4 e (u+\alpha+\beta)  \log n /\ell + O(u)$.
\end{theorem}

\begin{proof}
For $t\in\Nz$, let $\widehat{G}_{t}$ be the generalized directed graph obtained from $G_t$ by copying each edge $s-1$ times,
adding $r$ tailless edge at each vertex,
and adding $q$ headless edges at each vertex.
So $\widehat{G}_{t}$ has $s |E(G_t)| + (r+q) |V(G_t)|$ edges.
Note that the diameters of $G_t$ and $\widehat{G}_{t}$ are the same.
We claim that $(\widehat{G}_t)_{t=0}^{\infty}$ grows as described in Model~\ref{def_directed}.
First, sampling a vertex of $G_t$ using $\rho_{\beta}^{out}$ 
or
$\rho_{\alpha}^{in}$ 
correspond to choosing the tail or the head of a uniformly random tailed or headed edge of $\widehat{G}_t$, respectively.
Second, applying operation (a)  corresponds to applying only a vertex operation  with $A_{t}=sX_{t},B_{t}=0,C_t=q,D_t=r$.
Third, applying operation (b)  corresponds to applying only a vertex operation  with 
$A_{t}=0,B_t=sX_{t},C_t=q,D_t=r$.
Fourth, applying operation (c)  corresponds to applying only an edge operation  with 
$E_t=sX_t$. 
By Theorem~\ref{thm:directed}, a.a.s\ the diameter of $\widehat{G}_{n}$ is at most 
$4 e (u+\alpha+\beta)  \log n /\ell + 8e(su+q+r)+O(1)$, 
completing the proof.
\end{proof}

\section{The Cooper-Frieze model: multi-typed edge trees}
\label{combine}
In this section we study an undirected model that combines uniform and preferential attachment when choosing the neighbours of a new vertex.

\begin{model}
\label{def_cf}
Let $p_a,\dots,p_f$ be nonnegative numbers summing to 1
and satisfying $p_a+p_b>0$,
and let $(X_t)_{t\in\N}$ be a sequence of $\N$-valued random variables.
We consider a growing undirected graph $(G_t)_{t=0}^{\infty}$ as follows.
$G_0$ is an arbitrary connected graph.
At each time-step $t\in\N$, we perform exactly one of the following six operations, with probabilities $p_a,\dots,p_f$ and independently of previous choices.
\begin{enumerate} [(a)]
\item
$X_t$ vertices are sampled uniformly,
then a new vertex is born and is joined to the sampled vertices.
\item
$X_t$ vertices are sampled using $\rho_0$,
then a new vertex is born and is joined to the sampled vertices.
\item
$X_t+1$ vertices are sampled uniformly.
Then $X_t$ edges are added joining the first sampled vertex to the others.
\item
A vertex is sampled uniformly
and $X_t$ vertices are sampled using $\rho_0$.
Then $X_t$ edges are added joining the first sampled vertex to the others.
\item
A vertex is sampled using $\rho_0$
and $X_t$ vertices are sampled uniformly.
Then $X_t$ edges are added joining the first sampled vertex to the others.
\item
$X_t+1$ vertices are sampled using $\rho_0$.
Then $X_t$ edges are added joining the first sampled vertex to the others.
\end{enumerate}
Note that each operation increases the number of edges by $X_t$.
Again, we do not require any independence for $(X_t)_{t\in\N}$. 
\end{model}

Model~\ref{def_cf} is a generalization of a model defined by Cooper and Frieze~\cite[Section~2]{cooper-frieze-model},
in which the random variables $X_t$ are bounded.
The following theorem implies that a.a.s.\ the diameter of the latter model is $O(\log n)$.

\begin{theorem}
\label{thm:cooper-frieze}
Let $q = p_a + p_b$ and let $\ell=\ell(n),u=u(n)$ be positive integers such that $\ell \le X_t \le u$ for all $t$.
A.a.s.\ the diameter of $G_n$ generated by Model~\ref{def_cf} is at most 
$4(u/\ell+11/q) e\log n + 8e(u/\ell)+O(1)$.
\end{theorem}

\begin{proof}
As before, we define a growing tree whose height multiplied by 2 dominates the height of $(G_t)$, and then we upper bound the tree's height. The main difference with Theorem~\ref{thm:pref} is that in some operations we may sample the \emph{vertices} of the graph.
In a growing tree, when a new vertex $v$ is born and is joined to a vertex $w$ of the existing tree, we say
$w$ is the \emph{parent} of $v$, and that $w$ is \emph{given birth} to $v$.

We inductively define a growing tree $(T_t)_{t=0}^{\infty}$ such that $V(T_t) = V(G_t)\cup E(G_t)$ for all $t$,
and we prove that $\depth(f,G_t)\le2\depth(f,T_t)$ for each vertex or edge $f$ of $G_t$.
A node of $T_t$ is called a \emph{V-node} or an \emph{E-node}
if it corresponds to a vertex or an edge of $G_t$, respectively.
We may assume $T_0$ has been defined
(for instance, we can build it by taking a breadth-first search tree of $G_0$ and joining all the E-nodes to its deepest V-node)
and we describe the growth of $T_{t-1}$ to $T_t$ corresponding to each operation.

\begin{enumerate} [(a)]
\item
Let $w$ be the first sampled vertex.
In $T_t$ we join all new nodes (corresponding to the new vertex and the new edges in $G_t$) to $w$.
In this case, a V-node of $T_{t-1}$ has been sampled uniformly and is given birth to one V-node and $X_t$ E-nodes.

\item
For sampling a vertex using $\rho_0$,
we sample a random edge and then choose a random endpoint of it.
Let $e$ be the first sampled edge.
In $T_t$ we join all new nodes (corresponding to the new vertex and the new edges in $G_t$) to $e$.
In this case, an E-node of $T_{t-1}$ has been sampled uniformly and is given birth to one V-node and $X_t$ E-nodes.

\item[(c) and (d)]
Let $w$ be the first sampled vertex.
In $T_t$ we join all new nodes
(corresponding to the new edges  in $G_t$) to $w$.
In this case, a V-node of $T_{t-1}$ has been sampled uniformly and is given birth to $X_t$ E-nodes.

\item[(e) and (f)]
For sampling a vertex using $\rho_0$,
we sample a random edge and then choose a random endpoint of it.
Let $e$ be the first sampled edge.
In $T_t$ we join all new nodes
(corresponding to the new edges  in $G_t$) to $e$.
In this case, an E-node of $T_{t-1}$ has been sampled uniformly and is given birth to $X_t$ E-nodes.
\end{enumerate}

Similar to the proof of Theorem~\ref{thm:pref},
an inductive argument gives $\depth(f,G_t)\le2\depth(f,T_t)$ for each vertex or edge $f$ of $G_t$.
Hence, showing that a.a.s.\ the height of $T_n$ is at most $(u/\ell+11/q) e\log n + 2e(u/\ell)+O(1)$
completes the proof.

For $t\in\Nz$, let $L(t)$ denote the number of V-nodes of $T_t$.
Let $n_0=|V(T_0)|$ and $m_0=(9/q)\log n$.
Note that $L(t) = n_0 + \bin(t,q)$.
Using the Chernoff bound and the union bound, a.a.s\ we have $L(t) \ge tq/2$ for all $m_0 \le t \le n$.
We condition on an arbitrary vector $(L(1),\dots,L(n))=(g(1),\dots,g(n))$
for which this event happens.

For a given integer $h = h(n)$, we bound the probability that $T_n$ has a vertex at depth exactly $n_0+h$.
Given a sequence $1 \le  t_1 < \dots < t_h\le n$,
the probability that there exists a path $v_{t_1}v_{t_2}\dots v_{t_h}$ in $T_n$ such that $v_{t_j}$ is born at time $t_j$ is at most
$$
\prod_{k=2}^{h}\left(\frac{u}{n_0+ \ell\cdot(t_k-1)} +\frac{1}{g(t_k-1)}\right)\:,
$$
since for each $k=h,h-1,\dots,3,2$, 
if $v_{t_k}$ wants to choose an E-node as its parent,
there are at least $n_0+ \ell\cdot(t_k-1)$ E-nodes available for it to join to,
and at most $u$ of them were born at time $t_{k-1}$;
and if $v_{t_k}$ wants to choose a V-node as its parent,
there are at least $g(t_k-1)$ V-nodes available for it to join to,
and at most one of them was born at time $t_{k-1}$.
By the union bound, the probability that $T_n$ has a vertex at depth $h+n_0$ is at most
\begin{equation}
\label{eq:prob}
\sum_{1 \le t_1 < \dots < t_h\le n}\left( \prod_{k=2}^{h}\left(\frac{u}{n_0+ \ell\cdot(t_k-1)} +\frac{1}{g(t_k-1)}\right)\right)
 < \frac{1}{h!}
\left(1 + \sum_{j=1}^{n-1}\frac{u}{n_0+\ell j}+
\sum_{j=1}^{n-1}\frac{1}{g(j)}\right)^h \:.
\end{equation}
We have
$$
\sum_{j=1}^{n-1}\frac{u}{n_0+\ell j} 
< \frac{u}{\ell}
\sum_{j=1}^{n-1}\frac{1}{j} 
< (u/\ell)(1 + \log n)\:,
$$
and
$$
\sum_{j=1}^{n-1}\frac{1}{g(j)}
=
\sum_{j=1}^{m_0-1}\frac{1}{g(j)}
+
\sum_{j=m_0}^{n-1}\frac{1}{g(j)}
\le
m_0 + 
\sum_{j=m_0}^{n-1}\frac{2}{qj}
<
\frac{11}{q}\:\log n \:.
$$
Setting $h\ge(u/\ell+11/q) e\log n + 2e(u/\ell)$ makes the right hand side of (\ref{eq:prob}) 
become $o(1)$, as required.
\end{proof}

\section{Further models}
\label{sec:other}

We first mention a model whose diameter is known to be logarithmic, but our approach gives a shorter proof.
The \emph{pegging process}, defined by Gao and Wormald~\cite{pegging_def}, is parametrized by $d\in \N$.
Here we define the process for $d=3$ only, see~\cite[Section~2]{pegging_def} for the definition for $d>3$.
Consider a growing undirected graph $(G_t)_{t=0}^{\infty}$ starting from a connected $3$-regular $G_0$ and growing as follows. 
In every time-step, a pair $(e,f)$ of distinct edges is sampled uniformly from the existing graph.
Assume that $e=ab$ and $f=cd$.
Then two new vertices $e'$ and $f'$ are born,
the edges $e$ and $f$ are deleted,
and the edges $ae', be', cf', df'$, and $e'f'$ are added.
Note that if the original graph is 3-regular then the new graph is also 3-regular,
so $G_t$ is $3$-regular for all $t$.
Gerke, Steger, and Wormald~\cite[Theorem~1.1]{pegging} proved that for every $d$, a.a.s.\ $G_n$ has diameter $O(\log n)$.
Using the techniques of Section~\ref{sec:pref}, it can be shown that for every $d$, a.a.s.\ its diameter is at most $4e \log n + O(1)$.
We give the proof for the case $d=3$ here, which is much shorter than the 5-page proof in~\cite{pegging}, and provides a small explicit constant.

\begin{theorem}
Let $d\ge 3$ be fixed.
A.a.s.\ the diameter of the graph $G_n$ generated by pegging process is at most $4e\log n + O(1)$.
\end{theorem}

\begin{proof}
We give the proof for $d=3$, and the proof can easily be extended to $d>3$.
Define the \emph{depth} of an edge $xy$ as $1+\min\{\depth(x),\depth(y)\}$.
We inductively define a growing tree $(T_t)_{t=0}^{\infty}$ such that for all $t\in\Nz$,
$V(T_t) = E(G_t) \cup \{\root\}$.
Here $\root$ denotes the root of $T_t$, which has depth 0.
We prove by induction that for all $e\in E(G_t)$,
$\depth(e, G_t) \le 2 \depth(e,T_t)$.
Let $H$ be the graph obtained from $G_0$ by adding an edge labelled $\root$ incident to its root.
Let $T_0$ be a breadth-first tree of the line graph of $H$ rooted at $\root$. 
Note that $\depth(\root,T_0)=0$ and 
$\depth(e,T_0)=\depth(e,G_0)$ for every $e\in E(G_0)$.

Given $T_{t-1}$, we define $T_t$ and prove the inductive step.
Assume that in time-step $t$, the pair $(e,f)=(ab,cd)$ of edges is chosen from $G_{t-1}$.
By symmetry, we may assume that 
$\depth(a,G_{t-1}) \le \depth(b,G_{t-1})$
and
$\depth(d,G_{t-1}) \le \depth(c,G_{t-1})$.
Then two new vertices $e'$ and $f'$ are born,
the edges $ab$ and $cd$ are deleted,
and the edges $ae', be', cf', df'$, and $e'f'$ are added.
To obtain $T_t$ from $T_{t-1}$, we replace the nodes $ab$ and $cd$ with $ae'$ and $df'$, respectively.
Also, we join the other new edges $be'$, $e'f'$, and $cf'$ to $ae'$.
Observe that 
\begin{gather*}
\depth(ae',G_t) \le 1 + \depth(a,G_t) 
 = \depth(ab,G_{t-1}) \le 2 \depth(ab,T_{t-1})
= 2\depth(ae',T_{t}),\\
\depth(df',G_t) \le 1 + \depth(d,G_t) 
 = \depth(cd,G_{t-1}) \le 2 \depth(cd,T_{t-1})
= 2\depth(df',T_{t}),\\
\depth(be',G_t)  \le 1 + \depth(ae',G_t)
\le 1 + 2 \depth(ae', T_t)
< 2 \depth(be',T_t) \:,\\
\depth(f'e',G_t) \le 1 + \depth(ae',G_t)
 \le 1 + 2 \depth(ae', T_t)
< 2 \depth(f'e',T_t) \:,\\
\depth(cf',G_t)  \le 2 + \depth(ae',G_t)
 \le 2 + 2 \depth(ae', T_t)
= 2 \depth(cf',T_t) \:.
\end{gather*}
Hence for all $e\in E(G_t)$, $\depth(e, G_t) \le 2 \depth(e,T_t)$. On the other hand, examining the construction of $(T_t)_{t\in\Nz}$ and using Lemma~\ref{lem:multichildren}, we find that a.a.s.\ the height of $T_n$ is at most $e\log n +O(1)$.
This implies that a.a.s.\ the diameter of $G_n$ is at most $4 e\log n +O(1)$.
\end{proof}

Next we mention two closely related models for which we can easily prove logarithmic bounds using our technique.
Let $k>1$ be a positive integer.
A \emph{random unordered increasing $k$-tree}, defined by Gao~\cite{randomktree_def},
is built from a $k$-clique by applying the following operation $n$ times:
in every time-step, a $k$-clique of the existing graph is chosen uniformly at random, a new vertex is born and is joined to all vertices of the chosen $k$-clique.\footnote{The resulting graph is  named a \emph{random $k$-tree} in ~\cite{randomktree_def}. However, since a different model for generating $k$-trees has been defined in~\cite{CU10} and is also called a random $k$-tree, we used the name `random unordered increasing $k$-tree' here to avoid any confusion. 
This terminology is from~\cite{ktreenames}.}
\emph{Random $k$-Apollonian networks}~\cite{high_RANs} have a similar construction, the only difference being that once a $k$-clique is chosen in some time-step, it will never be chosen in the future.
Cooper and Frieze~\cite[Theorem~2]{CF13} 
and independently, Kolossv\'{a}ry, Komj\'{a}ty and V\'{a}g\'{o}~\cite[Theorem~2.2]{istvan}
have recently proved that random $k$-Apollonian networks have diameter $\Theta(\log n)$.

Here we prove that a.a.s.\ the diameter of a
random unordered increasing $k$-tree is at most $2e \log n + O(1)$,
and that a.a.s.\ the diameter of a
random $k$-Apollonian network is at most $2ek \log n/(k-1) + O(1)$.
For the proof for random $k$-Apollonian networks
we need the following variant of Lemma~\ref{lem:multichildren}.

\begin{lemma}
\label{lem:multichildren2}
For a tree $T$, let $\mathcal{L}(T)$ denote its set of leaves.
Let $(A_t)_{t\in\N}$ be a sequence of $\N$-valued  random variables.
Consider a growing tree $(T_t)_{t=0}^{\infty}$ as follows.
$T_0$ is arbitrary.
At each time-step $t\in\N$, 
a random vector $(W_1,W_2,\dots,W_{A_t}) \in \mathcal{L}(T_{t-1})^{A_t}$
is chosen in such a way that for each $i\in[A_t]$
and each $v\in \mathcal{L}(T_{t-1})$,
the marginal probability $\p{W_i=v}$ equals $|\mathcal{L}(T_{t-1})|^{-1}$. 
In other words, each $W_i$ is a leaf of $T_{t-1}$ sampled uniformly;
however, the $W_j$'s may be correlated.
Then $A_t$ new nodes $v_1,\dots,v_{A_t}$ are born and $v_i$ is joined to $W_i$ for each $i\in[A_t]$.
Let $\ell=\ell(n)$ and $u=u(n)$ be positive integers such that $1<\ell\le A_t \le u$ for all $t\in[n]$.
Then the height of $T_n$ is a.a.s.\ at most $ u e\log n /(\ell-1) + 2ue+O(1)$.
\end{lemma}

\begin{proof}
Let $n_0=|V(T_0)|$.
For a given integer $h = h(n)$, let us bound the probability that $T_n$ has a node at depth exactly $h+n_0$.
Given a sequence $1 \le t_1 < t_2<\dots < t_h\le n$,
the probability that there exists a path $v_{t_1}v_{t_2}\dots v_{t_h}$ in $T_n$ such that $v_{t_j}$ is born at time $t_j$ is at most
$$
u^{h} \prod_{k=2}^{h}\frac{1}{n_0+ (\ell-1)(t_k-1)} \:,
$$
since for each $k=2,3,\dots,h$, when $v_{t_k}$ is born, there are at least $n_0+ (\ell-1)(t_k-1)$ leaves available for it to join to.
By the union bound, the probability that $T_n$ has a node at depth $h+n_0$ is at most
\begin{align*}
u^{h}\sum_{1 \le t_1 < \dots < t_h\le n}\ \left( \prod_{k=2}^{h}\frac{1}{n_0+ (\ell-1)(t_k-1)}\right)
& < \frac{u^{h}}{h!}\left(1 + \sum_{j=1}^{n-1}\frac{1}{n_0+(\ell-1) j}\right)^{h} \\
& < \left(\frac{u e}{h} 
\cdot \left( \frac{\log n}{\ell-1} + 2\right) \right)^h
\bigg/\sqrt{2\pi h} \:.
\end{align*}
Putting $h \ge u e\log n /(\ell-1)+ 2ue$ makes this probability $o(1)$.
Hence a.a.s.\ the height of $T_n$ is at most 
$ u e\log n /(\ell-1) + 2ue + n_0$, as required.
\end{proof}

\begin{theorem}
A.a.s.\ the diameter of an 
$(n+k)$-vertex random unordered increasing $k$-tree is at most $2e \log n + O(1)$, and the diameter of an
$(n+k)$-vertex random $k$-Apollonian network is at most $2ek \log n/(k-1) + O(1)$.
\end{theorem}

\begin{proof}
We define the \emph{depth} of a $k$-clique as the maximum depth of its vertices.
Let the first $k$ vertices have depth zero.
We couple with a growing tree whose nodes corresponds to the $k$-cliques of the growing graph.
Whenever in the graph a new vertex is born and is joined to the vertices of a $k$-clique,
in the tree 
the chosen $k$-clique gives birth to $k$ new children.
By induction, the graph's height is always less than or equal to the tree's height.

For the tree corresponding to a random unordered increasing $k$-tree, in every step a node is chosen uniformly at random and gives birth to $k$ new children,
hence its height is bounded by $e \log n + O(1)$
by Lemma~\ref{lem:multichildren}.
This gives an upper bound of $2e \log n + O(1)$ for the diameter of the corresponding graph.

For the tree corresponding to a random $k$-Apollonian network, in every step a leaf is chosen uniformly at random and gives birth to $k$ new children,
hence its height is bounded by $e k \log n /(k-1) + O(1)$
by Lemma~\ref{lem:multichildren2}.
This gives an upper bound of $2e k\log n /(k-1) + O(1)$ for the diameter of the corresponding graph.
\end{proof}

%
%


\bibliographystyle{plain}
\bibliography{webgraph}	

\end{document}